\documentclass[journal,transmag]{IEEEtran}

\usepackage{mathrsfs}
\usepackage[dvips]{graphicx}
\usepackage{amsfonts}
\usepackage{amsthm}
\usepackage[cmex10]{amsmath}
\usepackage{hyperref}
\usepackage{algorithm}
\usepackage{algorithmic}
\usepackage{abstract}
\usepackage{bm} 
\usepackage{rotating}
\usepackage{array}
\usepackage{makecell}
\usepackage{float}
\usepackage{multirow}
\usepackage{caption}
\usepackage[inline]{enumitem}
\usepackage{cite}

\newtheorem{theorem}{Theorem}
\newtheorem{lemma}{Lemma}
\newtheorem{remark}{Remark}
\newtheorem{corollary}{Corollary}

\begin{document}

\title{Generalized Ensemble Model for Document Ranking in Information Retrieval}

\author{\IEEEauthorblockN{Yanshan Wang\IEEEauthorrefmark{1},
In-Chan Choi\IEEEauthorrefmark{2},
Hongfang Liu\IEEEauthorrefmark{1}}
\IEEEauthorblockA{\IEEEauthorrefmark{1}Department of Health Sciences Research,
        Mayo Clinic, Rochester, MN 55905, USA}
\IEEEauthorblockA{\IEEEauthorrefmark{2}School of Industrial Management Engineering,
       Korea University, Seoul 136-701, South Korea}
\thanks{Manuscript received ; revised .
Corresponding author: Y. Wang (email: Wang.Yanshan@mayo.edu).}}

\markboth{Journal of ,~Vol.~XX, No.~X, XX}%
{Shell \MakeLowercase{\textit{et al.}}: Bare Demo of IEEEtran.cls for Journals}

\IEEEtitleabstractindextext{%
\begin{abstract}
A generalized ensemble model (gEnM) for document ranking is proposed in this paper. The gEnM linearly combines basis document retrieval models and tries to retrieve relevant documents at high positions. In order to obtain the optimal linear combination of multiple document retrieval models or rankers, an optimization program is formulated by directly maximizing the mean average precision. Both supervised and unsupervised learning algorithms are presented to solve this program. For the supervised scheme, two approaches are considered based on the data setting, namely batch and online setting. In the batch setting, we propose a revised Newton's algorithm, gEnM.BAT, by approximating the derivative and Hessian matrix. In the online setting, we advocate a stochastic gradient descent (SGD) based algorithm---gEnM.ON. As for the unsupervised scheme, an unsupervised ensemble model (UnsEnM) by iteratively co-learning from each constituent ranker is presented. Experimental study on benchmark data sets verifies the effectiveness of the proposed algorithms. Therefore, with appropriate algorithms, the gEnM is a viable option in diverse practical information retrieval applications.
\end{abstract}

}

\maketitle

\IEEEdisplaynontitleabstractindextext

%
\IEEEpeerreviewmaketitle

\section{Introduction}

Ranking is a core task for Information Retrieval (IR) in practical applications such as search engines and advertising recommendation systems. The aim of ranking task is to retrieve the most relevant objects (documents, for example) with regard to a given query. With the continuous growth of information in modern world wide webs, this task has become more challenging than ever before. In the ranking task, the general problem is the over-inclusion of relevant documents that a user is willing to receive \cite{salton1986introduction}. During the last decade, a large quantity of models has been proposed to solve this problem. In general, those models are evaluated by two IR performance measures, namely Mean Average Precision (MAP) and Normalized Discounted Cumulative Gain (NDCG) \cite{jarvelin2000ir}. Compared to the framework in which models are proposed and then tested by IR measures, the approaches of directly optimizing IR measures have been showing more effective \cite{qin2010general, xu2008directly}. These approaches apply efficient algorithms to solve the optimization problem where the objective function is one of the IR measures.

Structured SVM is a widely used framework for optimizing the bound of IR measures. Examples include $\text{SVM}^{\textit{map}}$ \cite{yue2007support} and $\text{SVM}^{\textit{ndcg}}$ \cite{chapelle2007large}. Many other methods, such as Softrank \cite{taylor2008softrank,guiver2008learning}, first approximate the ranking measures through smooth functions and then optimize the surrogate objective functions. Yet, the drawbacks of those methods has been shown in two aspects:
\begin{enumerate*}[label=\itshape\alph*\upshape)]
\item the relationship between the surrogate objective functions and ranking measures was not sufficiently studied; and
\item the algorithms resolving the optimization problems are not trivial to be employed in practice \cite{qin2010general}.
\end{enumerate*}
Recently, a general framework that directly optimizes of IR measure has been reported \cite{qin2010general}. This framework can effectively overcome those drawbacks. However, it only optimizes the IR measure of one ranker, and the information provided by other rankers is not fully utilized.

In classification area, an ensemble classifier that linearly combines multiple classifiers has been successfully proved to perform better than any of the constituent classifiers. A number of sophisticated algorithms have been proposed for obtaining the ensemble classifier such as AdaBoost \cite{freund1995desicion}. Thus, the hypothesis that the performance can be improved by combining multiple rankers may be true as well. As a matter of fact, AdaRank \cite{xu2007adarank,wu2010adapting} and LambdaMART are two well-known models in IR area utilizing AdaBoost. The AdaRank repeatedly constructs weak rankers (features) and finally linearly combines into a strong ranker with proper weights assigned to the constituent rankers. However, the drawback of the AdaRank is the inexplicit theoretical justification and determination of the iteration number. While the LambdaMART enjoys the theoretical advantage of directly optimizing IR measures by linearly combining any two rankers, it cannot be extended to multiple rankers straightforwardly. In those previous studies, the direct optimization of NDCG is well-studied but the direct optimization of MAP are rarely tackled, to the best of our knowledge. The main difficulty of directly optimizing MAP is that the objective function defined by MAP is nonsmooth, nondifferentiable and nonconvex. Ensemble Model (EnM) \cite{wang2015indexing} solves this problem by using boosting algorithm and coordinate descent algorithm. However, the solutions cannot be theoretically guaranteed to be optimal, or even local optimal.

In this paper, we propose a generalized ensemble model (gEnM) for document ranking. It is an ensemble ranker that linearly combines multiple rankers. By appropriate adjustments to the weights for those constituent rankers, one may improve the overall performance of document ranking. To compute the weights, we formulate a constrained nonlinear program which directly optimizes the MAP. The difficulty of solving this nonlinear program lies in the nondifferentiable and noncontinuous objective function. To overcome this difficulty, we first introduce a differentiable surrogate to approximate the objective function, and then formulate an approximated unconstrained nonlinear program.

Both supervised and unsupervised algorithms are employed for solving the nonlinear program. In the supervised scheme, batch and online data settings are considered. These schemes and settings are designed for different IR environments. For the batch setting, the algorithm gEnM.BAT is a revised Newton's method by approximating the derivative and Hessian matrix. As for the online scheme, an online algorithm, gEnM.ON, is proposed based on stochastic gradient descent algorithms. The gEnM.ON is the first online algorithm for obtaining an ensemble ranker, to the best of our knowledge. In the unsupervised scheme, an unsupervised gEnM (UnsEnM) inspired by iRANK\cite{wei2010irank} is proposed. The UnsEnM utilizes the collaborative information among constituent rankers. The advantage of UnsEnM over the iRANK is that it is applicable to any number of constituent rankers. Compared to the EnM, the generalized version gEnM differs in three aspects:
\begin{enumerate}
  \item The assumption for EnM is relaxed for gEnM;
  \item the batch algorithms proposed for gEnM performs better;
  \item both online algorithm and unsupervised algorithm are proposed for gEnM whereas only batch algorithm for EnM.
\end{enumerate}

The remainder of this paper is organized as follows. In the next section, the problem of direct optimization of MAP is described and formulated. Also, the approximation to this problem is provided as long as the theoretical proofs. The algorithms, including gEnM.BAT, gEnM.ON and UnsEnM, are presented in Section 5. The computational results of the proposed algorithms tested on the public data sets are demonstrated in Section 6. The last section concludes this paper with discussions.

\section{Generalized Ensemble Model}
\subsection{Problem Description}

Consider the task of constructing a linear combination of rankers that result in better performance than each constituent. We call this linear combination the \textit{ensemble ranker} or \textit{ensemble model} hereinafter. Given a search query in this task, a sequence of documents is retrieved by the constituent rankers according to the relevance to the query. The relevance is measured by the ranking scores calculated by each ranker. For explicit description, let $\mathbf{score}_k$ denote the \textit{ranking score} or \textit{relevant score} calculated by the $k^{th}$ ranker. With appropriate weights $weight_k$ over those constituent rankers, the ranking scores $\mathbf{score}$ of ensemble ranker is defined by linearly summing the weighted constituent ranking scores, i.e., 
\begin{equation*}
\begin{aligned}
\mathbf{score}=&weight_1\cdot\mathbf{score}_1+weight_2\cdot\mathbf{score}_2+ \\
&\cdots+weight_k\cdot\mathbf{score}_k
\end{aligned}
\end{equation*}
where the weights satisfy $weight_i\geq0$ and $weight_1+weight_2+\cdots+weight_k=1$. The documents ranked by the ensemble ranker are thus ordered according to the ensemble ranker scores. Our goal is to uncover an optimal weight vector $$\mathbf{weight}=(weight_1,weight_2,...,weight_k)^T$$ with which more relevant documents can be ranked at high positions.

A toy example shown in Table \ref{tab:table.1} describes this problem. According to the ranking scores, the ranking lists returned by Ranker 1 and 2 are \{2,1,3\} and \{3,1,2\}, respectively, and the corresponding MAPs are 0.72 and 0.72. In order to make full use of the ranking information provided by both rankers, a conventional heuristic is to sum up ranking scores (i.e., use uniform weights, $(0.5,0.5)$), which generates Ensemble 1 with MAP equal to 0.72. Obviously, this procedure is not optimal since we can give arbitrary alternative weights that generate a better precision. For example, Ensemble 2 uses weights $(0.7,0.3)$ so as to result in higher MAP, i.e., 0.89, as listed in the table.

\begin{table}[h!]
\centering
\caption{A toy example. The values in the mid-three rows represent the ranking scores given an identical query. The rankers are measured by MAP, as listed in the fifth row. The ranking scores of Ensemble 1 and 2 are defined by 0.5*Ranker 1+0.5*Ranker 2 and 0.7*Ranker 1+0.3*Ranker 2, respectively. The relevant document list is assumed to be \{2,3\}. }
\label{tab:table.1}
\begin{tabular}{ccccc}
  \hline
   & Ranker 1 & Ranker 2 & Ensemble 1 & Ensemble 2\\
  \hline
  Document 1 & 0.35 & 0.2 & 0.55 & 0.305\\
  Document 2 & 0.4 & 0.1 & 0.5 & 0.31\\
  Document 3 & 0.25 & 0.7 & 0.95 & 0.385\\
  \hline
  MAP & 0.72 & 0.72 & 0.72 & 0.89\\
  \hline
\end{tabular}
\end{table}

This toy example implies that there exist optimal weights assigned for the constituent rankers to construct an ensemble ranker. Different from proposing new probabilistic or nonprobabilistic models, this ensemble model motivates an alternative way for solving ranking tasks. In order to formulate this task as an optimization problem, the metric---MAP---is used as the objective function since it reflects the performance of IR system and tends to discriminate stably among systems compared to other IR metrics \cite{robertson2012smoothing}. Therefore, our goal is changed to calculate the weights with which the MAP is maximized. In the following, we will describe and solve this problem mathematically.

\subsection{Problem Definition}

Let $D$ be a set of documents, $Q$ a set of queries and $\Phi$ a set of rankers. $|D_i|$ denotes the relevant document list, $d_j\in D$ the $d_j^{th}$ document associated with $j^{th}$ relevant document in $D_i$, $q_i\in Q$ the $i^{th}$ query and $\phi_k\in \Phi$ the $k^{th}$ ranker. $L$ represents the number of queries, $|D_i|$ the number of relevant documents associated with $q_i$ and $K_\phi$ the number of rankers. The ensemble ranker is defined as $H=\sum^{K_\phi}_{k=1}\alpha_k\phi_k$ which linearly combines $K_{\phi}$ constituent rankers with weights $\alpha$'s. We assume the relevant documents have been sorted in descending order according to the ranking sores. On the basis of these notations and the definition of MAP, the aforementioned problem can be formulated as:
\begin{equation*}\label{prob.orig}
  \begin{aligned}
    & \max
    & &\frac{1}{L}\sum^L_{i}\frac{1}{|D_i|}\sum^{|D_i|}_{j}\frac{j}
    {R\left(d_j,H\right)}\\
    & \text{s.t.}
    & &\sum^{K_\phi}_{k=1}\alpha_k=1\\
    &
    & &0\leq\alpha_k\leq1, k=1,2,...,K_{\phi}
  \end{aligned}\eqno(P1)
\end{equation*}
where $R\left(d_j,H\right)$ represents the ranking position of document $d_j$ given by the ensemble model $H$. In this constrained nonlinear program, \begin{enumerate*}[label=\itshape\alph*\upshape)]
\item the objective function is a general definition of MAP; and
\item the constraints indicate that the linear combination is convex and that the weights can be interpreted as a distribution.
\end{enumerate*}
Since the position function $R(d_j,H)$ is defined by the ranking scores, it can be written as
\begin{equation}\label{equ.r.orig}
  R(d_j,H)=1+\sum_{d\in D, d\neq d_j}\mathbf{I}\left\{s_{d_j,d}(H)<0\right\}
\end{equation}
where $s_{x,y}(H)=s_x(H)-s_y(H)$ and $\mathbf{I}\{s_{x,y}(H)<0\}$ is an indicator function which equals 1 if $s_{x,y}(H)<0$ is true and 0 otherwise. Here, $s_x(H)$ denotes the ranking score of document $x$ given by ensemble model $H$ and $s_{x,y}(H)$ the difference of the ranking scores between document $x$ and $y$. Since $s_x(H)$ is linear with respect to the weights, it can be rewritten as
\begin{equation}
\begin{aligned}
  s_x(H)&=s_x\left(\sum^{K_\phi}_{k=1}\alpha_k\phi_k(q_i)\right) \\
  &=\sum^{K_\phi}_{k=1}\alpha_ks_x(\phi_k(q_i))
\end{aligned}
\end{equation}
where $s_x(\phi_k(q_i))$ denotes the relevant score of document $x$ for query $q_i$ calculated by model $\phi_k$.

Here, we give an example plot that illustrates the graph of the objective function. This example employed the MED data set with the settings identical to those in \cite{wang2015indexing} except that only two constituent rankers, LDI and pLSI, were used to comprise the ensemble ranker for plotting purpose. The weights were restricted to the constraints in Problem P1 with the precision of three digits after the decimal point. In detail, the objective function was evaluated by setting $\alpha_1$ for LDI and $\alpha_2$ for pLSI, where $\alpha_1+\alpha_2=1$, and $\alpha_1$ increased from $0$ to $1$ with a step size of $0.001$. Figure \ref{fig.zoom} shows a partial of the graph of objective function. From this plot, it is clearly observed that \begin{enumerate*}[label=\itshape\alph*\upshape)]
\item the objective function is highly nonsmooth and nonconvex; and
\item there are numerous local optimums in the objective function.
\end{enumerate*}
Though the differentiability is not obvious in this graph, the position function implies that the objective function is nondifferentiable in terms of weights. Therefore, the general gradient-based algorithms, such as Lagrangian Relaxation and Newton's Method, cannot be applied to this problem directly to find the optimum, even local optimums \cite{qin2010general}.

\begin{figure}[h!]
  \centering
  \includegraphics[width=0.5\textwidth]{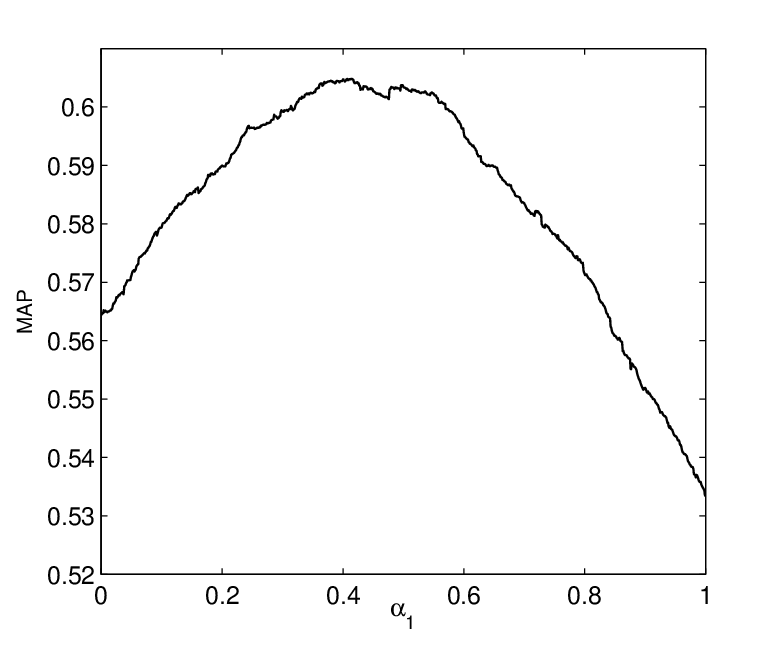}
  \caption{An illustrated example of the objective function with two constituent rankers in Problem P1. }\label{fig.zoom}
\end{figure}

From this analysis of the objective function, the position function plays an important role in the differentiability. Thus, we will discuss how to approximate it with a differentiable function and how to solve this optimization Problem P1 in the next two sections.

\section{Approximation}
In this section, we propose a differentiable surrogate for the position function and further approximate the Problem P1 with an easier nonlinear program.

Since the position function is defined by an indicator function (Equation \ref{equ.r.orig}), we can use a sigmoid function to approximate this indicator function, i.e.,
\begin{equation}
  \mathbf{I}\{s_{d_j,d}(H)<0\}\simeq \frac{\exp(-\beta s_{d_j,d}(H))}{1+\exp(-\beta s_{d_j,d}(H))},
\end{equation}
where $\beta >0$ is a scaling constant. It is obvious that this approximation is in the range of $[0.5,1)$ if $s_{d_j,d}(H)\leq0$ and $(0,0.5]$ if $s_{d_j,d}(H)>0$. The following theorem shows that we can get a tight bound by this approximation.

\begin{theorem}\label{lemma.2}
The difference between the sigmoid function $g_{ij}$ and the indicator function $\mathbf{I}\{s_{d_j,d}(H)<0\}$ is bounded as:
$$\left|g_{ij}-\mathbf{I}\{s_{d_j,d}<0\}\right|
<\frac{1}{1+\exp(\beta\delta_{ij})}$$
where $\delta_{ij}=\min|s_{d_j,d}|$, $g_{ij}=\frac{\exp(-\beta \sum^{K_\phi}_{k=1}\alpha_k s_{d_j,d})}{1+\exp(-\beta \sum^{K_\phi}_{k=1}\alpha_k s_{d_j,d})}$ and $s_{d_j,d}$ represents $s_{d_j,d}(\phi_k(q_i))$ for notational simplicity henceforth
\end{theorem}

\begin{proof}
For $s_{d_j,d}>0$, we have $\mathbf{I}\{s_{d_j,d}<0\}=0$ and $\delta_{ij}\leq s_{d_j,d}$, thus,
\begin{equation*}
  \left|g_{ij}-\mathbf{I}\{s_{d_j,d}<0\}\right|\leq\frac{1}{1+\exp(\beta\delta_{ij}\sum^{K_\phi}_{k=1}\alpha_k)}
\end{equation*}
For $s_{d_j,d}<0$, we have $\mathbf{I}\{s_{d_j,d}<0\}=1$ and $\delta_{ij}\leq-s_{d_j,d}$, thus,
\begin{equation*}
\begin{aligned}
  &\left|g_{ij}-\mathbf{I}\{s_{d_j,d}<0\}\right| \\
  &\leq\frac{1}{1+\exp(\beta\delta_{ij}\sum^{K_\phi}_{k=1}\alpha_k)}
\end{aligned}
\end{equation*}
Since $\sum^{K_\phi}_{k=1}\alpha_k=1$, we can get
\begin{equation}
\left|g_{ij}-\mathbf{I}\{s_{d_j,d}<0\}\right|
\leq\frac{1}{1+\exp(\beta\delta_{ij})}.
\end{equation}

This completes the proof.
\end{proof}

This theorem tells us that the sigmoid function is asymptotic to the indicator function especially when $\beta$ is chosen to be large enough. By using this approximation, the position function can be correspondingly approximated as
\begin{equation}\label{equ.psap}
  \hat{R}(d_j,H)=1+\sum_{d\in D, d\neq d_j}\frac{\exp(-\beta s_{d_j,d}(H))}{1+\exp(-\beta s_{d_j,d}(H))},
\end{equation}
which becomes differentiable and continuous.

Then it is trivial to show the approximation error of position function, i.e.,
\begin{equation}\label{equ.R.err}
\begin{aligned}
  \left|\hat{R}(d_j,H)-R(d_j,H)\right|
  &\leq \sum_{d\in D, d\neq d_j}\left|g_{ij}-\mathbf{I}\{s_{d_j,d}<0\}\right| \\
  &<\frac{|D|-1}{1+\exp(\beta\delta_{ij})}.
\end{aligned}
\end{equation}

Suppose 1000 documents exit in the document set $D$ and $\delta_{ij}=0.04$. By setting $\beta=300$, the approximation error of the position function is bounded by
\begin{equation}
\left|\hat{R}(d_j,H)-R(d_j,H)\right|<0.006,
\end{equation}
which is tight enough for our problem.

In this way, the original Problem P1 can be approximated by the following problem
\begin{equation*}\label{prob.sec}
 \begin{aligned}
    & \max
    & &\frac{1}{L}\sum^L_{i=1}\frac{1}{|D_i|}\sum^{|D_i|}_{j=1}\frac{j}{\hat{R}(d_j,H)}\\
    & \text{s.t.}
    & &\sum^{K_\phi}_{k=1}\alpha_k=1\\
    &
    & &0\leq\alpha_i\leq1, i=1,2,...,K_{\phi}.
  \end{aligned}\eqno(P2)
  \end{equation*}

Using the settings identical to Figure \ref{fig.zoom}, Figure \ref{fig.apc} plots the graphs of the original objective function (OOF) in Problem P1 and the approximated objective function (AOF) in Problem P2. As shown in the plot, the trend of the AOF is close to that of the OOF. The weights generating the optimal MAP almost remain unchanged in these two curves. From this example, it is illustratively shown that the original noncontinuous and nondifferentiable objective function can be effectively approximated by a continuous and differentiable function. The following lemma and theorem will theoretically prove this conclusion.

\begin{figure}[h!]
  \centering
  \includegraphics[width=0.5\textwidth]{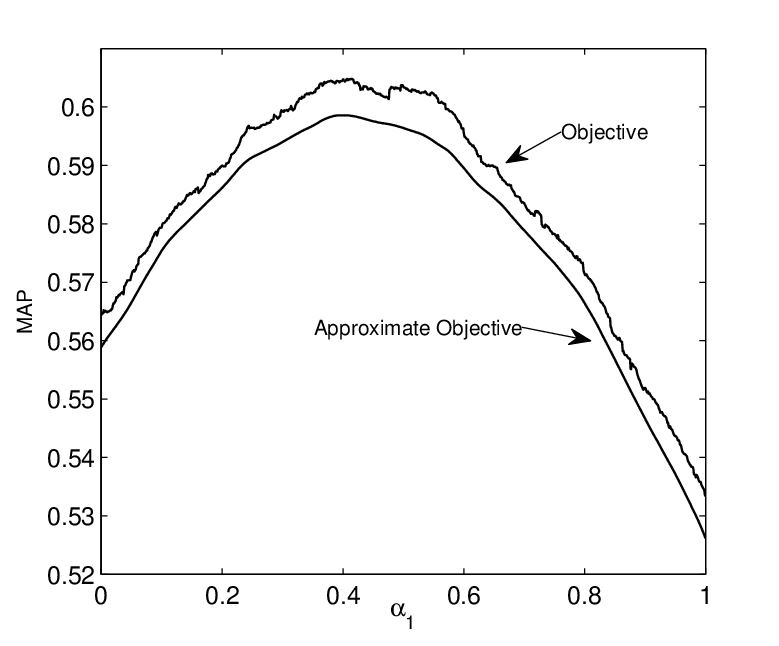}
  \caption{Comparison of the OOF in Problem P1 and AOF in Problem P2. ($\beta=200$)}\label{fig.apc}
\end{figure}

\begin{theorem}
\label{thm.1}
  The error between the OOF in Problem P1 and the AOF in Problem P2 is bounded as
  \begin{equation}\label{equ.g.sum}
    |\hat{\Lambda}-\Lambda|<\frac{(|D|-1)(L+\sum_i |D_i|)}{2L(1+\exp(\beta \delta_{ij}))}
  \end{equation}
  where $\hat{\Lambda}$ and $\Lambda$ denote the objective function in Problem P2 and Problem P1, respectively.
\end{theorem}

\begin{proof}
  For the approximation error, we have
  \begin{equation*}
  |\hat{\Lambda}-\Lambda|
  =\frac{1}{L}\sum^L_{i=1}\frac{1}{|D_i|}\sum^{|D_i|}_{j=1}\left|\frac{j(R-\hat{R})}{R\hat{R}}\right|,
  \end{equation*}
  where $R$ denotes $R(d_j,H)$ for notational simplicity.
  Since $\hat{R}=1+\sum_{d\neq d_j}g_{ij}(\alpha)$ and $R=1+\sum_{d\neq d_j}\mathbf{I}\{s_{d_j,d}<0\}$ are strictly positive, we have
  \begin{equation*}
  \begin{aligned}
  &\left|\frac{jj(R-\hat{R})}{R\hat{R}}\right|  \\
  &=\frac{j\left|R-\hat{R}\right|}{R\hat{R}} .\\
  \end{aligned}
  \end{equation*}

  According to Equation \ref{equ.R.err}, we have
  \begin{equation}
    |\hat{\Lambda}-\Lambda| <\frac{(|D|-1)(L+\sum_i |D_i|)}{2L(1+\exp(\beta \delta_{ij}))}.
  \end{equation}
  This completes the proof.
\end{proof}

This theorem indicates that the OOF in Problem P1 can be accurately approximated by the surrogate defined by the position function (\ref{equ.psap}) in Problem P2. For example, if $|D|=10000$, $L=200$, $\sum |D_i|=500$, $\beta=300$ and $\delta_{ij}=0.04$, the absolute discrepancy between the objectives in Problem P1 and P2 is bounded by
$$|\hat{\Lambda}-\Lambda|<0.1.$$
This discrepancy is within an acceptable level and will decrease with the growth of the query size $L$ and the value of $\beta$.

The constraints of weights in Problem P2 are of practical significance because these weights can be regarded as probabilities drawn from a distribution over the constituent rankers. However, adding constraints increases the difficulty of solving this optimization problem. Intuitively, the normalization of weights assigned for ranking scores is nonessential because the ranking position is determined by the relative values of ranking scores. Take the toy in Table \ref{tab:table.1} as an example, the weights $(3.5,1.5)$ result in the identical Ensemble 2 to $(0.7,0.3)$. The lemmas and theorems below prove the hypothesis that this constrained nonlinear program can be approximated by an unconstrained nonlinear program.

\begin{lemma}\label{lemma.1}
  Problem P2 is equivalent to the following problem:
  \begin{equation*}\label{prob.thi}
  \begin{aligned}
    &\mathrm{max}
    & &\frac{1}{L}\sum^L_{i=1}\frac{1}{|D_i|}\sum^{|D_i|}_{j=1}\frac{j}{\tilde{R}}\\
  \end{aligned}\eqno(P3)
  \end{equation*}
  where $\tilde{R}=1+\sum_{d\in D, d\neq d_j}\tilde{g}_{ij}$, $\tilde{g}_{ij}=\frac{\exp(-\beta \sum^{K_\phi}_{k=1}\tilde{\alpha}_k s_{d_j,d}(\phi_k(q_i)))}{1+\exp(-\beta \sum^{K_\phi}_{k=1}\tilde{\alpha}_k s_{d_j,d}(\phi_k(q_i)))}$ and $\tilde{\alpha}_k=\frac{\alpha_k'}{\sum^{K_\phi}_{k=1}\alpha_k'},\alpha_k'>0,k=1,2,...,K_{\phi}$
\end{lemma}

Since $\sum^{K_\phi}_{k=1}\tilde{\alpha}_k=1$, it can be straightforwardly proved that Problem P3 is equivalent to Problem P2.
\begin{remark}
If we let $g'_{ij}=\frac{\exp(-\beta \sum^{K_\phi}_{k=1}\alpha_k' s_{d_j,d}(\phi_k(q_i)))}{1+\exp(-\beta \sum^{K_\phi}_{k=1}\alpha_k' s_{d_j,d}(\phi_k(q_i)))}$, Theorem \ref{lemma.2} applies for both $\tilde{g}_{ij}$ and $g'_{ij}$ as well.
\end{remark}

The following theorem states that Problem P3 can be surrogated by an easier problem.

\begin{theorem}\label{thm.final}
  Consider the following problem
  \begin{equation*}\label{prob.noconst}
  \begin{aligned}
    &\mathrm{max}
    & &\frac{1}{L}\sum^L_{i=1}\frac{1}{|D_i|}\sum^{|D_i|}_{j=1}\frac{j}{R'}, \\
  \end{aligned}\eqno(P4)
  \end{equation*}
  where $R'=1+\sum_{d\in D, d\neq d_j}g'_{ij}$.
  Let $\tilde{\Lambda}$ and $\Lambda'$ denote the objective function in Problem P3 and Problem P4, respectively. Then, we have the following bound for the absolute difference between $\tilde{\Lambda}$ and $\Lambda '$
  \begin{equation}
    |\tilde{\Lambda}-\Lambda '|<\frac{\hat{\epsilon}(L+\sum_{i=1}^L |D_i|)}{2L}
  \end{equation}
  where $\hat{\epsilon}=\epsilon'+\tilde{\epsilon}$, $\epsilon'=\left|R'-R\right|$ and $\tilde{\epsilon}=\left|\tilde{R}-R\right|$.
\end{theorem}

\begin{proof}
  From Lemma \ref{lemma.1} and Lemma \ref{lemma.2}, we can derive the following bound.
  \begin{equation*}
  \begin{aligned}
  &|\tilde{\Lambda}-\Lambda '| \\
  &=\frac{1}{L}\sum^L_{i=1}\frac{1}{|D_i|}\sum^{|D_i|}_{j=1}\left|\frac{j(R'-\tilde{R})}{R'\tilde{R}}\right| \\
  \end{aligned}
  \end{equation*}
  Since $R'=1+\sum_{d\neq d_j}g'_{ij}$ and $\tilde{R}=1+\sum_{d\neq d_j}\tilde{g}_{ij}$ are strictly positive, we have
  \begin{equation*}
  \begin{aligned}
  &\left|\frac{j\left(\sum_{d\neq d_j}g'_{ij}-\sum_{d\neq d_j}\tilde{g}_{ij}\right)}{(1+\sum_{d\neq d_j}g'_{ij})(1+\sum_{d\neq d_j}\tilde{g}_{ij})}\right|  \\
  &=\frac{j\sum_{d\neq d_j}\left|(g'_{ij}-\mathbf{I}\{s_{d_j,d}<0\})+(\mathbf{I}\{s_{d_j,d}<0\}-\tilde{g}_{ij})\right|}{(1+\sum_{d\neq d_j}g'_{ij})(1+\sum_{d\neq d_j}\tilde{g}_{ij})} \\
  \end{aligned}
  \end{equation*}
  According to the general triangle inequality, we can draw an upper bound for the term in numerator
  \begin{equation*}
  \begin{aligned}
  &\sum_{d\neq d_j}\left|(g'_{ij}-\mathbf{I}\{s_{d_j,d}<0\})+(\mathbf{I}\{s_{d_j,d}<0\}-\tilde{g}_{ij})\right| \\
  &\leq\sum_{d\neq d_j}\left|g'_{ij}-\mathbf{I}\{s_{d_j,d}<0\}\right|+\sum_{d\neq d_j}\left|\mathbf{I}\{s_{d_j,d}<0\}-\tilde{g}_{ij}\right|
  \\
  &< \hat{\epsilon}. \\
  \end{aligned}
  \end{equation*}

Then, it is trivial to get
\begin{equation}
\begin{aligned}
  |\tilde{\Lambda}-\Lambda '|
  &< \frac{1}{L}\sum^L_{i=1}\frac{1}{|D_i|}\sum^{|D_i|}_{j=1}j\cdot \hat{\epsilon} \\
  &<\frac{\hat{\epsilon}(L+\sum_{i=1}^L |D_i|)}{2L}.
\end{aligned}
\end{equation}

This completes the proof.
\end{proof}

Since the differences $\epsilon'$ and $\tilde{\epsilon}$ are small enough, Problem P4 can accurately approximate Problem P3. This theorem tells us that the AOF is also determined by the ranking positions, i.e., the relative values of ranking scores, thus the normalization constraints in Problem P2 can be removed. Taking Lemma \ref{lemma.1} and Theorem \ref{thm.1} into account, we can trivially draw the following corollary.

\begin{corollary}
   Problem P1 can be approximated by Problem P4.
\end{corollary}

In the next section, we focus on proposing algorithms that solves Problem P4.

\section{Algorithm}\label{sec.alg}

In order to solve Problem P4, we propose algorithms according to the data settings---batch setting and online setting. In the batch setting, all the queries and ranking scores given by constituent rankers are processed as a batch. Based on the batch data, the weights over constituent rankers are computed by maximizing the MAP. Two algorithms, gEnM.BAT and gEnM.IP, are reported in this setting. The potential for the batch algorithms merit consideration for those systems containing complete data. Take academic search engine as an example. The titles can be seen as queries while the abstracts and contents of publications can be regarded as relevant documents. So a batch can be established to train the proposed model.

In many IR environments such as recommendation systems in E-commerce, however, the queries and ranking scores are generated in real time so as to construct data sequences at different times. Thus, we will secondly propose an online algorithm, gEnM.ON, for dealing with these data sequences. The online algorithm is more scalable to large data sets with limited storage than the batch algorithm. In the online algorithm, the queries as well as corresponding ranking scores are input in a data stream and processed in a serial fashion.

A common assumption for the aforementioned frameworks is that the relevant documents are known. However, the knowledge of relevant documents are unknown in many modern IR systems such as search engines. For this IR environment, we further propose an unsupervised ensemble model, UnsEnM, which makes use of a co-training framework.

\subsection{Batch Algorithm: gEnM.BAT}
Although many sophisticated methods can be applied for finding a local optimum, we first propose a revised Newton's method. Major modification includes the approximation of gradients and Hessian matrix.

For notational simplicity, we utilize:
\begin{equation}\label{equ.Gij}
  G_{ij}:=\sum_{d\in D,d\neq d_j}g'_{ij};
\end{equation}
\begin{equation}\label{equ.Gijk}
  G_{ij}^k:=\sum_{d\in D,d\neq d_j}\frac{\partial g'_{ij}}{\partial\alpha'_k};
\end{equation}
\begin{equation}\label{equ.Gijl}
  G_{ij}^l:=\sum_{d\in D,d\neq d_j}\frac{\partial g'_{ij}}{\partial\alpha'_l};
\end{equation}
\begin{equation}\label{equ.Gijkl}
  G_{ij}^{kl}:=\sum_{d\in D,d\neq d_j}\frac{\partial^2 g'_{ij}}{\partial\alpha'_k\partial\alpha'_l}.
\end{equation}

Under those notations, the first and second derivative of the objective function in Problem P4 can be written as
\begin{equation}
  \frac{\partial\Lambda'}{\partial\alpha'_k}=\frac{1}{L}\sum^L_{i=1}\frac{1}
  {|D_i|}\sum^{|D_i|}_{j=1}\frac{-jG_{ij}^k}{(1+G_{ij})^2},
\end{equation}
and
\begin{equation}
\begin{aligned}
\frac{\partial^2\Lambda'}{\partial\alpha'_k\partial\alpha'_l}=&\frac{1}{L}\sum^L_{i=1}\frac{1}
  {|D_i|} \\
  &\sum^{|D_i|}_{j=1}\frac{-jG_{ij}^{kl}(1+G_{ij})^2+2jG_{ij}^k G_{ij}^l(1+G_{ij})}{(1+G_{ij})^2},
\end{aligned}
\end{equation}
respectively. According to the second derivative, the Hessian matrix is defined by
\begin{equation}\label{equ.hess}
\mathcal{H}(\bm{\alpha})=\left[\begin{array}{cccc}
  \frac{\partial^2\Lambda'}{\partial\alpha'_1\partial\alpha'_1}& \frac{\partial^2\Lambda'}{\partial\alpha'_1\partial\alpha'_2}& \cdots & \frac{\partial^2\Lambda'}{\partial\alpha'_1\partial\alpha'_{K{\phi}}}\\
  \frac{\partial^2\Lambda'}{\partial\alpha'_2\partial\alpha'_1}& \frac{\partial^2\Lambda'}{\partial\alpha'_2\partial\alpha'_2}& \cdots & \frac{\partial^2\Lambda'}{\partial\alpha'_2\partial\alpha'_{K{\phi}}}\\
  \vdots & \vdots&  & \vdots\\
  \frac{\partial^2\Lambda'}{\partial\alpha'_{K{\phi}}\partial\alpha'_1}& \frac{\partial^2\Lambda'}{\partial\alpha'_{K{\phi}}\partial\alpha'_2}& \cdots & \frac{\partial^2\Lambda'}{\partial\alpha'_{K{\phi}}\partial\alpha'_{K{\phi}}}\\
\end{array}\right].
\end{equation}

As stated by Theorem \ref{thm.sig.app} in Appendix \ref{appendix.b}, the addends in the first derivative can be estimated by zeros under certain conditions. This approximation also applies for the second derivative as well as the Hessian matrix since both contain the first derivative item. The advantages of using this approximation are two-fold:
\begin{enumerate*}[label=\itshape\alph*\upshape)]
\item the computation of Hessian is simplified since many addends are set to zeros under certain conditions; and
\item the computations of $G_{ij}^{kj}$, $G_{ij}$, $G_{ij}^l$ and $G_{ij}^k$ can be carried out offline before evaluating the derivative and Hessian, which makes the learning algorithm inexpensive.
\end{enumerate*}

Since the objective function in Problem P4 is nonconvex, multiple local optimums may exist in the variable space. Therefore, different starting points are chosen to preclude the algorithm from getting stuck in one local optimum. The largest local optimum and the corresponding weights are returned as the final solutions. To accelerate the algorithm, we can distribute different starting points onto different cores for parallel computing.

The batch algorithm is summarized as follows. We note that $\bm{\alpha}_p$ and $\textbf{s}_{d_j,d}(\phi(q_i))$ represent the vectors with elements $\alpha_p$ and $s_{d_j,d}(\phi_k(q_i))$, respectively, and that $p=1,2,...,P$ indexes $P$ initial values.

\begin{algorithm}[h!]
  \caption{gEnM.BAT (Generalized Ensemble Model by Revised Newton's Algorithm in Batch Setting.)}
  \label{alg.1}
  \begin{algorithmic}[1]
  \REQUIRE Query set $Q$, document set $D$, relevant document set $|D_i|$ with respect to $q_i\in Q$, ranking scores $s_{d}(\phi_k(q_i))$ with respect to $i$the query, $k$th method $\phi_k$ and document $d\in D$, a number of initial points $\bm{\alpha}_p$ and a threshold $\epsilon=0$ for stopping the algorithm.
  \FOR{each $\bm{\alpha}_p$}
  \STATE Set iteration counter $t=1$;
  \STATE Evaluate $\Lambda'^t$;
  \REPEAT
    \STATE Set $t=t+1$;
    \STATE Compute gradient $\nabla_{\bm{\alpha}_p^{t-1}}\Lambda'$ and Hessian matrix $H(\bm{\alpha}_p^{t-1})$ (Algorithm \ref{alg.2});
    \STATE Update $\bm{\alpha}_p^{t}=\bm{\alpha}_p^{t-1}+\mathcal{H}(\bm{\alpha}_p^{t-1})^{-1}\nabla_{\bm{\alpha}_p^{t-1}}\Lambda'$;
    \STATE Evaluate $\Lambda'^t$;
  \UNTIL{$\Lambda'^{t}-\Lambda'^{t-1}<\epsilon$}
  \STATE Store $\bm{\alpha}_p^t$
  \ENDFOR
  \RETURN $\bm{\alpha}$'s.
  \end{algorithmic}
\end{algorithm}

\begin{algorithm}[h!]
  \caption{Approximated Derivative and Hessian Computation Algorithm.}
  \label{alg.2}
  \begin{algorithmic}[1]
  \REQUIRE Query set $Q$, document set $D$, relevant document set $|D_i|$ with respect to $q_i\in Q$, ranking scores $s_{d}(\phi_k(q_i))$ with respect to $i$the query, $k$th method $\phi_k$ and document $d\in D$, current $\bm{\alpha}_p^{t-1}$.
  \FOR{$q_i\in Q$}
  \FOR{$d_j\in |D_i|$}
  \STATE Set $G_{ij}$, $G_{ij}^{kl}$, $G_{ij}^k$ and $G_{ij}^l$ to zeros;
  \FOR{$d\in D$}
  \STATE $s_{d_j,d}(\phi_k(q_i))\leftarrow s_{d_j}(\phi_k(q_i))-s_{d}(\phi_k(q_i))$;
  \STATE $g'_{ij}(\bm{\alpha}_p^{t-1}) \leftarrow \frac{\exp(-\beta\bm{\alpha}_p^{t-1} \textbf{s}_{d_j,d}(\phi(q_i)))}{1+\exp(-\beta\bm{\alpha}_p^{t-1} \textbf{s}_{d_j,d}(\phi(q_i)))}$;
  \STATE $G_{ij} \leftarrow G_{ij}+g'_{ij}(\bm{\alpha}_p^{t-1})$
  \IF{$-\frac{2}{\beta}<\bm{\alpha}_p^{t-1} \textbf{s}_{d_j,d}(\phi(q_i))<\frac{2}{\beta}$}
    \STATE $G_{ij}^{kl} \leftarrow G_{ij}^{kl}+\beta^2 s_{d_j,d}(\phi_k(q_i)) s_{d_j,d}(\phi_l(q_i))g'_{ij}(\bm{\alpha}_p^{t-1})(1-g'_{ij}(\bm{\alpha}_p^{t-1}))(1-2g'_{ij}(\bm{\alpha}_p^{t-1}))$;
    \STATE $G_{ij}^k \leftarrow G_{ij}^k+\beta s_{d_j,d}(\phi_k(q_i))$;
    \STATE $G_{ij}^l \leftarrow G_{ij}^l+\beta s_{d_j,d}(\phi_l(q_i))$;
  \ELSE
    \STATE $G_{ij}^{kl} \leftarrow G_{ij}^{kl}$;
    \STATE $G_{ij}^k \leftarrow G_{ij}^k$;
    \STATE $G_{ij}^l \leftarrow G_{ij}^l$;
  \ENDIF
  \ENDFOR
  \ENDFOR
  \ENDFOR
  \STATE Compute gradient $\nabla_{\bm{\alpha}_p^{t-1}}\Lambda'$ \hfill(Equation \ref{equ.der.app}) \\
   and Hessian matrix $\mathcal{H}(\bm{\alpha}_p^{t-1})$;
   \hfill(Equation \ref{equ.hess})
  \RETURN $\nabla_{\bm{\alpha}_p^{t-1}}\Lambda'$ and $\mathcal{H}(\bm{\alpha}_p^{t-1})$.
  \end{algorithmic}
\end{algorithm}

A drawback of the conventional Newton's method lies in that it is designed for unconstrained nonlinear programs while our problem requests $\alpha$ nonnegative. Thus applying the above algorithms may result in negative weights. The strategy for avoiding this shortcoming is to set the final negative weights to zeros. As a matter of fact, the rankers with negative weights play a negative role in the ensemble model. Thus, the ignorance of those rankers are reasonable in practice.

\subsection{Online Algorithm: gEnM.ON}

In the previous two subsections, we have presented the learning algorithms for generating gEnM by batch data sets. In contrast to the batch setting, the online setting provides the gEnM a long sequence of data. The weights are calculated sequentially based on the data stream that consists of a series of time steps $t=1,2,...,T$. For example, the gEnM is constructed based on the new queries and corresponding rankings given at different times in a search engine. The final goal is also to maximize the overall MAP on the data sets.
\begin{equation}
\max \frac{1}{T}\sum_{t=1}^T \frac{1}{D_t}\sum^{D_t}_{j=1}\frac{j}{1+\sum_{d\in D, d\neq d_j}g'_{ij}}
\end{equation}

As a matter of fact, the presented batch algorithms can be applied directly in the online setting by regarding the whole observed sequences as a batch at each step. In doing so, however, the overall complexity is extremely high since the batch algorithm should be run once at each time step.

In the online setting, the subsequent queries are not available at present. An alternative optimization technique should be considered to prevent from focusing too much on the present training data. To distinguish with the notation in the batch setting, we let $\mathbf{x}$ be the query and suppose $\mathbf{x}_1,\mathbf{x}_2,...\mathbf{x}_t,...$ are the given query at time $t$ in the online setting. Here, we assume that these sequences are given with the \textit{grand truth distributio}n $p(\mathbf{x})$. Thus, the objective function of MAP can be defined as the expectation of average precision, i.e.,
\begin{equation}
\begin{aligned}
J(\alpha) &=\sum_{t=1}^{\infty} f(\mathbf{x},\alpha)p(\mathbf{x}) \\
      &=\mathrm{E}_p[f(\mathbf{x},\alpha)],
\end{aligned}
\end{equation}
where $$f(\mathbf{x},\alpha)=\frac{1}{D_{x_t}}\sum^{D_{x_t}}_{j=1}\frac{j}{1+\sum_{d\in D,d\neq d_j}g'_{{x_t}j}(\alpha')}.$$

The expectation cannot be maximized directly because the truth distribution $p(\mathbf{x})$ is unknown. However, we can estimate the expectation by the \textit{empirical MAP} that simply uses finite training observations. A plausible approach for solving this empirical MAP optimization problem is that using the stochastic gradient descent (SGD) algorithm which is a drastic simplification for the expensive gradient descent algorithm. Though the SGD algorithm is a less accurate optimization algorithm compared to the batch algorithm, it is faster in terms of computational time and cheaper in terms of storing memory \cite{murata1998statistical}. Another advantage is that the SGD algorithm is more adaptive to the changing environment in which examples are given sequentially \cite{amari1967theory}.

For our problem, the SGD learning rule is formulated as
\begin{equation}\label{equ.uprule}
\alpha_{t+1}=\alpha_t+\eta_t \nabla f(\mathbf{x}_{t+1},\alpha_t)
\end{equation}
where $\eta_t$ is called learning rate, i.e., a positive value depending on $t$. This updating rule is validated to increase the objective value at each step in terms of expectation, which can be verified by the following theorem.

\begin{theorem}
Using the updating rule (\ref{equ.uprule}), the expectation of average precision increases at each step, i.e.,
$$\mathrm{E}_p[f(\mathbf{x},\alpha_{t+1})]\geq \mathrm{E}_p[f(\mathbf{x},\alpha_t)]$$
\end{theorem}

\begin{proof}
Since $\mathrm{E}_p[f(\mathbf{x},\alpha_{t+1})]- \mathrm{E}_p[f(\mathbf{x},\alpha_t)]=\mathrm{E}_p[f(\mathbf{x},\alpha_{t+1})-f(\mathbf{x},\alpha_t)]$, we only need to show $f(\mathbf{x},\alpha_{t+1})-f(\mathbf{x},\alpha_t)\geq 0$. \\
Since
\begin{equation*}
\begin{aligned}
f(\mathbf{x},\alpha_{t+1})& -f(\mathbf{x},\alpha_t) = \frac{1}{D_{x}} \\
& \sum^{D_{x}}_{j=1} \left(\frac{j\sum_{d\neq d_j}(g'_{xj}(\alpha'_{t+1})-g'_{xj}(\alpha'_{t}))}{(1+\sum_{d\neq d_j}g'_{xj}(\alpha'_{t+1}))(1+\sum_{d\neq d_j}g'_{xj}(\alpha'_{t}))}\right),
\end{aligned}
\end{equation*}
we need to verify $g'_{xj}(\alpha'_{t+1})-g'_{xj}(\alpha'_{t})\geq 0$. According to the denotation of $g'_{ij}$, we have
$$g'_{xj}(\alpha'_{t+1})-g'_{xj}(\alpha'_{t}) = \frac{\tau(\alpha'_t)-\tau(\alpha'_{t+1})} {(1+\tau(\alpha'_t))(1+\tau(\alpha'_{t+1}))} $$
where $\tau(\alpha'_t)=\frac{g'_{xj}(\alpha'_{t})}{1-g'_{xj}(\alpha'_{t})}$. \\
Since
\begin{equation}\label{equ.eta.role}
\begin{aligned}
\frac{\tau(\alpha'_t)}{\tau(\alpha'_{t+1})} &=\exp(\beta\eta_t \nabla \mathbf{f}(\mathbf{x},\alpha'_t)\mathbf{s}(\phi)) \\
&\geq exp(0) \\
&=1,
\end{aligned}
\end{equation}
we can conclude that
$$\tau(\alpha'_t)-\tau(\alpha'_{t+1}) \geq 0.$$
This completes the proof.
\end{proof}

The learning rate $\eta$ plays an important role in the updating (Equation \ref{equ.eta.role}), hence an adequate $\eta_t$ will enhance the online algorithm to converge. Define $\eta_t=1/t$ in this article, then we have the following well-known properties:
\begin{equation}\label{equ.eta.con1}
\sum_t^\infty \eta_t^2<\infty,
\end{equation}
\begin{equation}\label{equ.eta.con2}
\sum_t^\infty \eta_t=\infty.
\end{equation}

Since it is difficult to analyze the whole process of online algorithm \cite{murata1998statistical}, we will show the convergence property around the global or local optimum in the following analysis.

\begin{lemma}\label{lemma.cond1}
If $\alpha_t$ is in the neighborhood of the optimum $\alpha^*$, we have
\begin{equation}
(\alpha_t-\alpha^*)\nabla f(\mathbf{x},\alpha_t)<0.
\end{equation}
\end{lemma}
The proof of is straightforward referring to Equation \ref{equ.g.fird}. This lemma states that the gradient drives the current point towards the maximum $\alpha^*$. In the stochastic process, the following inequality holds
\begin{equation}\label{equ.con.1}
(\alpha_t-\alpha^*)\mathrm{E}_p[\nabla f(\mathbf{x},\alpha_t)]<0.
\end{equation}

\begin{lemma}\label{lemma.cond2}
If $\alpha_t$ is in the neighborhood of the optimum $\alpha^*$, we have
\begin{equation}
\lim_{t\rightarrow \infty} \nabla f(\mathbf{x},\alpha_t)^2 < \infty.
\end{equation}
\end{lemma}

The proof is given in the Appendix. For the stochastic nature, the expectation of $\nabla f(\mathbf{x},\alpha_t)^2$ also converges almost surely, i.e.,
\begin{equation}\label{equ.con.2}
\lim_{t\rightarrow \infty} \mathrm{E}_p[\nabla f(\mathbf{x},\alpha_t)^2] < \infty.
\end{equation}

\begin{theorem}[\cite{bottou1998online}]
In the neighborhood of the maximum $\alpha^*$, the recursive variables $\alpha$ converge to the maximum, i.e.,
\begin{equation}
\lim_{t\rightarrow \infty}\alpha_t=\alpha^*.
\end{equation}
\end{theorem}

\begin{proof}
Define a sequence of positive numbers whose values measure the distance from the optimum, i.e.,
\begin{equation}
h_{t+1}-h_{t}=(\alpha_t-\alpha^*)^2.
\end{equation}
The sequence can be written as an expectation under the stochastic nature, i.e.,
\begin{equation}\label{equ.h.var}
\mathrm{E}_p[h_{t+1}-h_{t}]=2\eta_t(\alpha_t-\alpha^*) \mathrm{E}_p[\nabla f(\mathbf{x},\alpha_t)]+\eta_t^2 \mathrm{E}_p[\nabla f(\mathbf{x},\alpha)^2]
\end{equation}
Since the first term on the right hand side is negative according to (\ref{equ.con.1}), we can obtain the following bound:
\begin{equation}
\mathrm{E}_p[h_{t+1}-h_{t}]\leq \eta_t^2 \mathrm{E}_p[\nabla f(\mathbf{x},\alpha_t)^2].
\end{equation}
Conditions (\ref{equ.eta.con2}) and (\ref{equ.con.2}) imply that the right hand side converges. According to the quasi-martingale convergence theorem \cite{fisk1965quasi}, we can also verify that $h_t$ converges almost surely.  This result implies the convergence of the first term in (\ref{equ.h.var}).

Since $\sum_t^\infty \eta_t$ does not converge according to (\ref{equ.eta.con1}), we can get
\begin{equation}
\lim_{t\rightarrow \infty}(\alpha_t-\alpha^*) \mathrm{E}_p[\nabla f(\mathbf{x},\alpha_t)]=0.
\end{equation}
This result leads to the convergence of the online algorithm, i.e.,
\begin{equation*}
\lim_{t\rightarrow \infty}\alpha_t=\alpha^*.
\end{equation*}
This completes the proof.
\end{proof}

Based on the learning rule (\ref{equ.uprule}), the online algorithm for achieving the ensemble model is summarized below.
\begin{algorithm}[h!]
  \caption{gEnM.ON (Generalized Ensemble Model by Online Algorithm.)}
  \label{alg.5}
  \begin{algorithmic}[1]
  \REQUIRE Query set $Q$, document set $D$, relevant document set $|D_i|$ with respect to $q_i\in Q$, ranking scores $s_{d}(\phi_k(q_i))$ with respect to $i$the query, $k$th method $\phi_k$ and document $d\in D$, a number of initial points $\bm{\alpha}_p$ and a threshold $\epsilon>0$ for stopping the algorithm.
  \FOR{each $\bm{\alpha}_p$}
  \STATE Set iteration counter $t=1$;
  \STATE Evaluate $\Lambda'^t$;
  \REPEAT
  \FOR{each $q_i\in Q$}
    \STATE Set $t=t+1$;
    \STATE Compute gradient $\nabla_{\bm{\alpha}_p^{t-1}}\Lambda'$ with respect to $q_i$ \hfill(Algorithm \ref{alg.2});
    \STATE Update $\bm{\alpha}_p^{t}=\bm{\alpha}_p^{t-1}+\frac{1}{t}\nabla_{\bm{\alpha}_p^{t-1}}\Lambda'$;
  \ENDFOR
  \STATE Evaluate $\Lambda'^t$;
  \UNTIL{$|\Lambda'^{t}-\Lambda'^{t-1}|<\epsilon$}
  \STATE Store $\bm{\alpha}_p^t$
  \ENDFOR
  \RETURN $\bm{\alpha}$'s.
  \end{algorithmic}
\end{algorithm}

\subsection{Unsupervised Algorithm: UnsEnM}

The proceeding proposed algorithms for both batch setting and online setting are based on the knowledge of labeled data, which has been regarded as supervised learning. As a matter of fact, in the community of conventional information retrieval systems, labeled data are difficult to obtain in general. Under this condition, unsupervised learning plays a crucial role. The inspiration of unsupervised algorithm for solving Problem P4 comes from the idea of co-training that is based on the belief that each constituent ranker in the ensemble model can provide valuable information to the other constituent rankers such that they can co-learn from each other \cite{wei2010irank}. In order to utilize this collaborative learning scheme, the gEnM requires all constituent rankers are generated by unsupervised learning. In each round, the ranking scores of one of the constituent rankers are provided as \textit{fake} labeled data for other rankers to refine the weights. Iteratively learning from the constituent rankers, the ensemble model may result in an overall improvement in terms of MAP.

We modify the objective function in Problem P4 by adding a penalty item so that the refined ranking does not depend on the fake label too much. The modified objective function is defined as
\begin{equation*}\label{prob.unsup}
  \begin{aligned}
    &\mathrm{max}
    & &\Lambda'-\frac{1}{2}\sigma\sum_{q_i \in Q}\sum_{d \in D}\sum_{\phi_k\in \Phi}\left\|H_d(q_i)-s_d(\phi_k(q_i))\right\|^2\\
  \end{aligned}\eqno(P8)
  \end{equation*}
where $H_d(q_i)=\sum_k^{k\in K_{\phi}}\alpha_k s_d(\phi_k(q_i))$. \\

Let $\Gamma$ denote the objective function in Problem P8. The second derivatives of $\Gamma$ can be written as follows:
\begin{equation}
\frac{\partial\Gamma}{\partial\alpha_k\alpha_l}=\frac{\partial^2\Lambda'}{\partial\alpha_k\alpha_l}
-\sigma\sum_{q_i \in Q}\sum_{d \in D}\left(s_d(\phi_k(q_i))\cdot s_d(\phi_l(q_i))\right)
\end{equation}
The approximation of Hessian matrix reported in Algorithm 2 can be employed here, however, it is time-consuming doing so since the unsupervised algorithm requires a large number of iterations to converge and the Hessian should be calculated at each iteration. Therefore, the learning rule of the online algorithm gEnM.ON is applied for the unsupervised algorithm. It is noteworthy that the gEnM.ON can be effortlessly modified to fit this unsupervised co-training scheme. The algorithm is described below.

\begin{algorithm}[h!]
  \caption{UnsEnM (\underline{Uns}upervised \underline{En}semble \underline{M}odel.)}
  \label{alg.6}
  \begin{algorithmic}[1]
  \REQUIRE Query set $Q$, document set $D$, ranking scores $s_{d}(\phi_k(q_i))$ with respect to $i$the query, $k$th method $\phi_k$ and document $d\in D$, a number of initial points $\bm{\alpha}_p$, a threshold $\epsilon_s$ for $s_{d}(\phi_k(q_i))$ to choose fake relevant documents and a threshold $\epsilon>0$ for stopping the algorithm.
  \FOR{each $\bm{\alpha}_p$}
  \STATE Set iteration counter $t=1$;
  \STATE Evaluate $\Lambda'^t$;
  \REPEAT
  \FOR{each $\phi_k\in \Phi$}
    \STATE Set $t=t+1$;
    \STATE Refresh fake relevant document set $|D_i|=\emptyset$;
    \STATE Construct $\hat{s}_{d}$ that excludes $s_{d}(\phi_k)$;
    \STATE Construct $\bm{\alpha}_p$ that excludes $\alpha_{\phi_k}$;
    \FOR{$q_i\in Q$}
        \IF{$s_{d}(\phi_k(q_i))>\epsilon_s$}
            \STATE Construct fake relevant document set $|D_i|\leftarrow i\cup |D_i|$;
        \ENDIF
    \ENDFOR
    \STATE Compute gradient $\nabla_{\bm{\alpha}_p^{t-1}}\Lambda'$; \hfill(Algorithm \ref{alg.2})
    \STATE Update $\bm{\alpha}_p^{t}=\bm{\alpha}_p^{t-1}+\frac{1}{t}\nabla_{\bm{\alpha}_p^{t-1}}\Lambda'$;
  \ENDFOR
  \STATE Reconstruct $\bm{\alpha}_p$ that includes $\alpha_{\phi_k}$;
  \STATE Evaluate $\Lambda'^{t}$;
  \UNTIL{$|\Lambda'^{t}-\Lambda'^{t-1}|<\epsilon$}
  \STATE Store $\bm{\alpha}_p^t$
  \ENDFOR
  \RETURN $\bm{\alpha}$'s.
  \end{algorithmic}
\end{algorithm}

\section{Empirical Experiment}

\subsection{Experiment Setup}
The proposed methods were evaluated on four standard medium-sized ad-hoc document collections, i.e., MED, CRAN, CISI and CACM, which can be accessed freely from the SMART IR System\footnote{Available at: ftp://ftp.cs.cornell.edu/pub/smart.}. In order to test the proposed methods on heterogeneous data, we utilized the merged collection (MC) advocated by \cite{wang2015indexing}, which combines the four collections. The basic statistics of the test data are summarized in Table \ref{tab.med}. The following minimum pre-processing measures were taken for the collections before evaluating the proposed methods:
\begin{enumerate*}[label=\itshape\alph*\upshape)]
  \item stop words were removed from the corpus by referring to a list of 571 stop words provided by SMART\footnotemark[1];
  \item special symbols, including hyphenation marks, were removed; and
  \item those words with unique appearances in the corpus were removed.
\end{enumerate*}
We note that the incomplete documents and queries in CISI and CACM were retained in the experiments.

\begin{table}[h!]
\centering
\caption{Data characteristics.}
\label{tab.med}
\begin{tabular}{ccccc}
  \hline
  Data & Subject & Document \# & Query \# & Term \# \\
  \hline
  MED & Medicine & 1,033 & 30 & 5,775 \\
  CRAN & Aeronautics & 1,400 & 225 & 8,213 \\
  CISI & Library & 1,460 & 112 & 10,170 \\
  CACM & Computer & 3,204 & 64 & 9,961 \\
  MC & Multiplicity & 7,097 & 431 & 27,784 \\
  \hline
\end{tabular}
\end{table}

The constituent rankers, in essence, are important factors that influence the results. Four rankers recommended by \cite{wang2015indexing}, namely \textit{tf-idf}-based ranker (TFIDF) \cite{salton1986introduction}, Latent Semantic Analysis (LSA) \cite{deerwester1990indexing}, probabilistic Latent Semantic Indexing (pLSI) \cite{hofmann1999probabilistic}, Indexing by Latent Dirichlet Allocation (LDI) \cite{wang2015indexing}, were utilized in this paper for assembling the gEnM. In brief, TFIDF represents documents by a tf-idf weighted matrix; LSA projects each document into a lower dimensional conceptual space by applying Singular Value Decomposition (SVD); pLSI is a probabilistic version of LSA; and LDI represents each document by a probabilistic distribution over shared topics based on Latent Dirichlet Allocation (LDA) \cite{blei2003latent}. These rankers are all unsupervised rankers and thus are trivial to be trained in the unsupervised setting. In addition to this training requirement, the rankers contain different information describing each corpus, such as information of keyword matching, concepts, or topics.

Since the four rankers represent documents and queries into vectors, the ranking scores are the cosine distances (or cosine similarities) between the vectors of documents and queries. Subsequently, the ranking scores of gEnM can be generated with appropriate adjustments to the weights being made for the ranking scores of the four rankers. For formulating Problem P4, we set $\beta=200$. Finally, the proposed algorithms can be implemented to calculate the optimal weights for gEnM.

In order to address the over-fitting problem of batch algorithms, we adopted the two-fold cross validation for testing the gEnM.BAT and gEnM.ON. A difference for the gEnM.ON is that the training queries and corresponding relevant documents were given sequentially at each step. The performance metric was the mean value of the MAPs in the two-fold cross validation. As for the UnsEnM, the ranking scores of different constituent rankers are provided as labeled data for other rankers in different rounds. The UnsEnM was then evaluated by means of MAP on the real labeled data.

As discussed in Section \ref{sec.alg}, the proposed algorithms would benefit from different initial weights. Choosing the proper initial points for nonlinear program is an open research issue. In our tests, we utilized the operational criterion of selecting the best. In other words, we tested performances for different initial weights and selected the one that generated the maximum retrieval performance in terms of MAP. In this experiment, we first set the initial weights to binary elements, i.e., $\bm{\alpha}\in \mathbf{B}^4$. The reason of doing so lies in that the constituent rankers are initially active in some of the rankers and inactive in others, which reflects our heuristics at the first step. Since the EnM has been shown prior to the four basis rankers by \cite{wang2015indexing}, the EnM model was used as baseline methods for comparison.

\subsection{Experimental Results}

The experimental results are shown in Table \ref{tab.map}. We have considered three measures for comparing the performances of the proposed algorithms: mean average precision (MAP), (average) precision at one document (Pr@1), and (average) precision at five documents (Pr@5). Indeed, the gEnM performance is always better than the EnM. Since the EnM is also solved by a batch algorithm, we conduct the Wilcoxon signed rank test to evaluate the difference between EnM and gEnM.BAT. We see that, in some cases, the difference is statistically significant with a 95\% confidence. We emphasize that the Pr@1 of gEnM is 48\% higher than that of EnM for the CISI data set and is close to 100\% for the MED. In other words, the retrieved documents by gEnM are more relevant at high ranking positions, which is desirable from the user's point of view.

\begin{table*}
\centering
\caption{Comparison of the algorithms for gEnM and baseline methods. Pr@1 denotes the precision at one document and Pr@5 the precision at five documents. An asterisk (*) indicates a statistically significant difference between EnM and gEnM.BAT with a 95\% confidence according to the Wilcoxon signed rank test.}
\label{tab.map}
\begin{tabular}{l|c|ccccl}
  \hline
   Collection & Measure & EnM & gEnM.BAT & gEnM.ON & UnsEnM & impr(\%) \\
  \hline
  \multirow{3}{*}{MED} & MAP & 0.6420 & 0.6458 & \textbf{0.6467} & 0.6465 & +0.6 \\
  & Pr@1 & 0.8667 & 0.9333 & \textbf{0.9333} & 0.9333 & +7.7*\\
  & Pr@5 & 0.7867 & 0.8133 & \textbf{0.8133} & 0.8133 & +3.4*\\
  \hline
  \multirow{3}{*}{CRAN} & MAP & 0.3766 & 0.3937 & \textbf{0.3972} & 0.3972 & +4.5\\
  & Pr@1 & 0.6133 & 0.6622 & \textbf{0.6667} & 0.6356 & +8.0*\\
  & Pr@5 & 0.3742 & \textbf{0.4080} & 0.3991 & 0.4018 & +9.0*\\
  \hline
  \multirow{3}{*}{CISI} & MAP & 0.1637 & \textbf{0.1945} & 0.1816 & 0.1825 & +18.8*\\
  & Pr@1 & 0.3289 & \textbf{0.4868} & 0.3684 & 0.3947 & +48.0* \\
  & Pr@5 & 0.2974 & \textbf{0.3237} & 0.2868 & 0.3079 & +8.8 \\
  \hline
  \multirow{3}{*}{CACM} & MAP & 0.1890 & 0.2166 & \textbf{0.2256} & 0.1745 & +14.6* \\
  & Pr@1 & 0.3654 & 0.3846 & \textbf{0.4423} & 0.3077 & +5.3 \\
  & Pr@5 & 0.2192 & 0.2500 & \textbf{0.2538} & 0.2000 & +14.1* \\
  \hline
  \multirow{3}{*}{MC} & MAP & 0.2768 & 0.3162 & 0.3099 & \textbf{0.3169} & +14.2* \\
  & Pr@1 & 0.4204 & 0.5196 & 0.5300 & \textbf{0.5274} & +23.6* \\
  & Pr@5 & 0.307 & 0.3614 & 0.3624 & \textbf{0.3629} & +17.7* \\
  \hline
\end{tabular}
\end{table*}




From Table \ref{tab.map}, we also see that the performance of gEnM.ON is better than the gEnM.BAT. The slight priority of gEnM.ON is due to the approximation of Hessian for the gEnM.BAT. However, the gEnM.ON is more expensive than gEnM.BAT because of iterative use of queries for calculation. Having said that, gEnM.ON can be used in a specific system where data are given in sequence. Since the knowledge of relevant documents is unknown in unsupervised learning, the performance of UnsEnM is inferior to the supervised algorithms. However, the results on the more heterogeneous data set MC are surprisingly the best among the proposed algorithms. The supervised algorithm may work well when tested against similar queries and documents in the homogeneous data. Yet the unsupervised algorithm does not fit the training data as much as the supervised algorithm does and thus the superiority becomes more obvious when tested on more heterogeneous data.

Figure \ref{fig.pr_4} shows the precision-recall curves of the examined methods.

\begin{figure}[h!]
  \centering
  \includegraphics[width=0.5\textwidth]{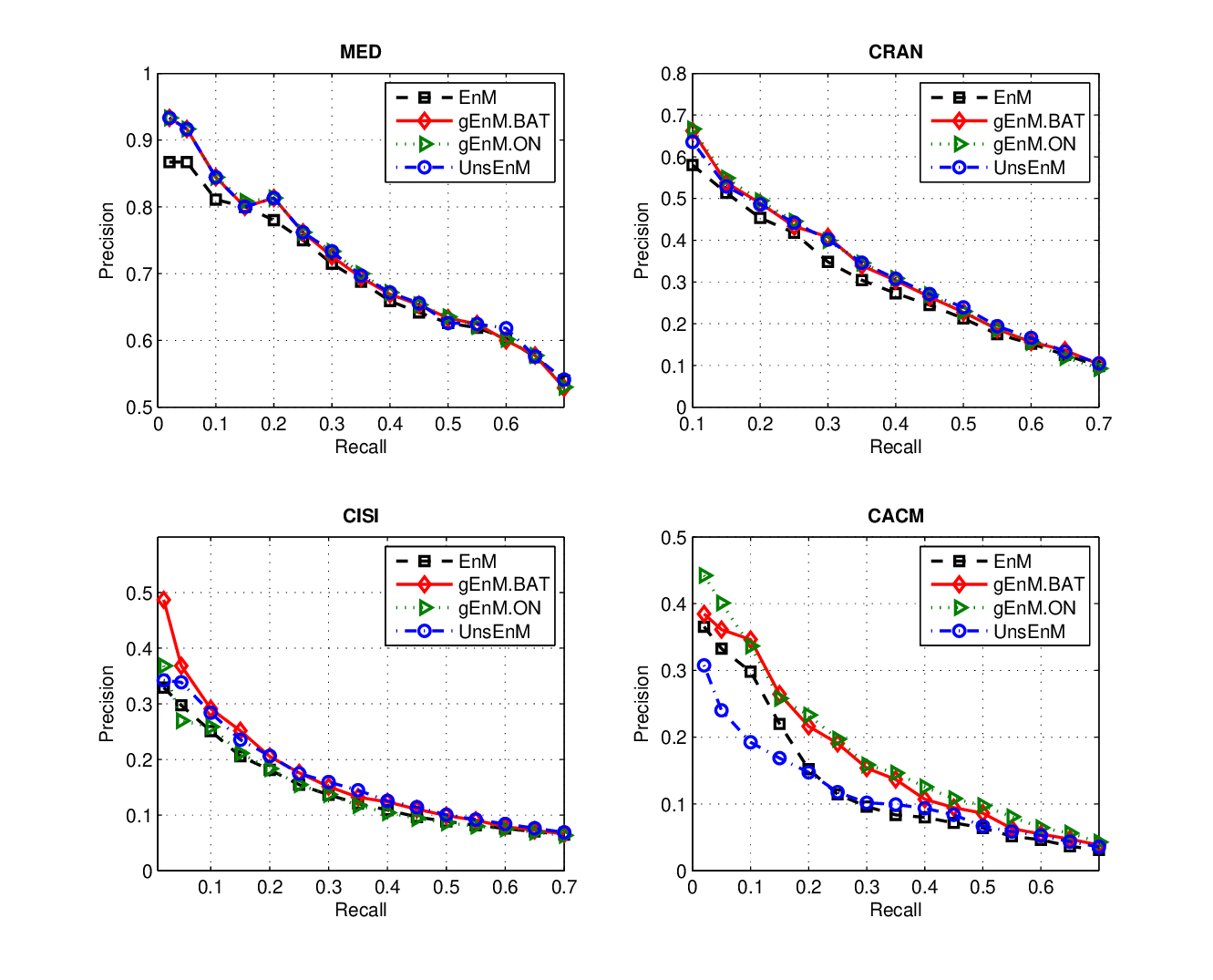}
  \caption{Precision-Recall Curves for the testing data sets.}\label{fig.pr_4}
\end{figure}

\begin{figure}[h!]
\ContinuedFloat
  \centering
  \includegraphics[width=0.5\textwidth]{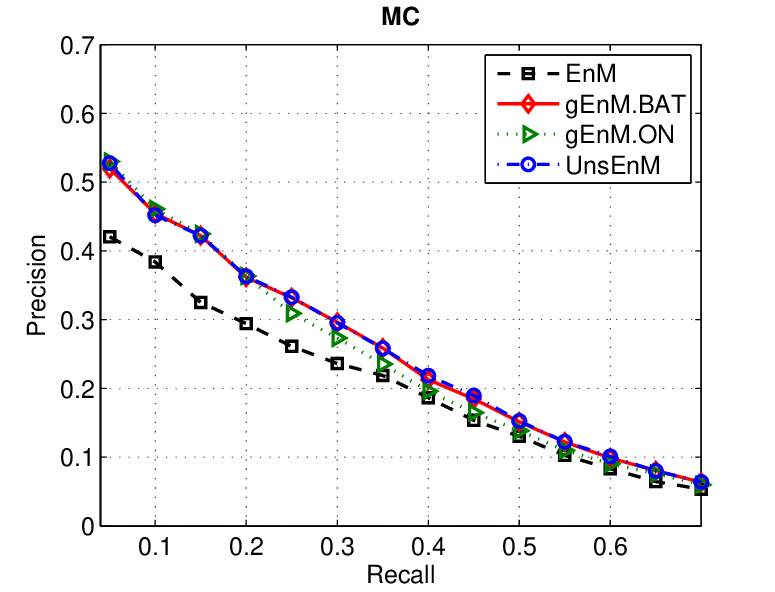}
  \caption[]{Precision-Recall Curves for the testing data sets. (continued)}\label{fig.pr_mc}
\end{figure}

For illustrating the learning abilities of the gEnM.ON and UnsEnM, the learning curves on the MED data are reported in Figure \ref{fig.l.curv1}. The results on the other data sets are very similar. The tolerance is set to $1e-4$ and the number of iteration is set to at least $10$ in order to clearly view the changes of objective. The online learning curves validates the convergence property of gEnM.ON. Amongst these curves, several scenarios, such as when $\alpha=(1,1,1,1)^T$ and $\alpha=(1,0,0,0)^T$, imply that the gEnM.ON may occasionally fail for some queries that are not similar to the previous sequences and not near the local optimum. With the increase of iterations, however, the impact of those queries may mitigate due to the majority effect. Apart from these specific cases, the gEnM.ON is able to gradually learn from the sequences, which is consistent with the theoretical analysis.

The UnsEnM also converges with the increase of iterations. We can see that in the case of $\alpha=(1,0,0,0)^T$ a ranker which is regarded as supervised labels may dramatically decrease the objective function. In most cases, the impact of such rankers can be balanced out by other rankers. As a matter of fact, this phenomenon is similar to gEnM.ON since the data are given sequentially in both cases.

\newsavebox{\mybox}
\newcolumntype{X}[1]{%
>{\begin{lrbox}{\mybox}}%
c%
<{\end{lrbox}\makecell[#1]{\usebox\mybox}}%
}

\begin{sidewaysfigure*}\tiny
\centering
  \begin{tabular}{X{cc}X{cc}X{cc}X{cc}X{cc}}
  \hline
  Initial $\alpha$ & (0;0;0;0) & (1;1;1;1) & (1;0;0;0) & (0;1;0;0) \\
  \hline
gEnM.ON &\includegraphics[width=4.5cm]{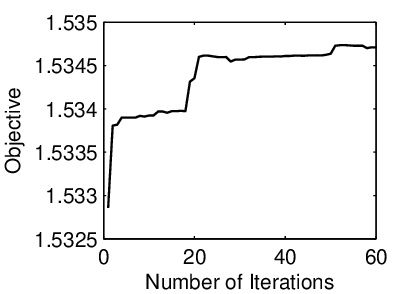}& \includegraphics[width=4.5cm]{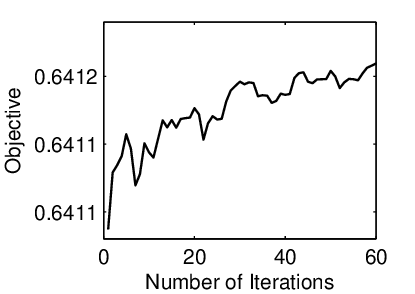} & \includegraphics[width=4.5cm]{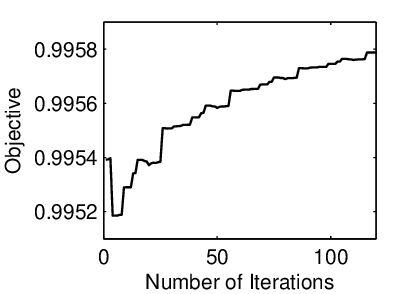} & \includegraphics[width=4.5cm]{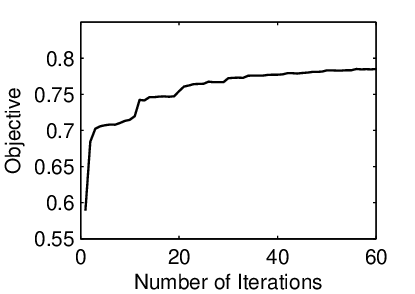}\\
UnsEnM &\includegraphics[width=4.5cm]{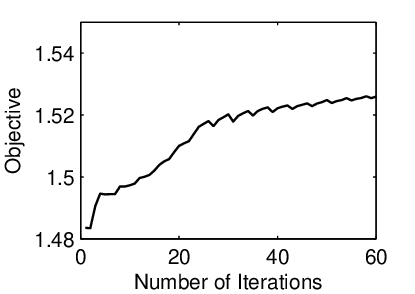}& \includegraphics[width=4.5cm]{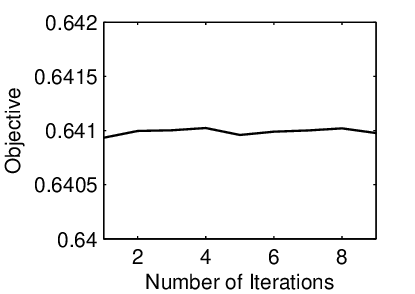} & \includegraphics[width=4.5cm]{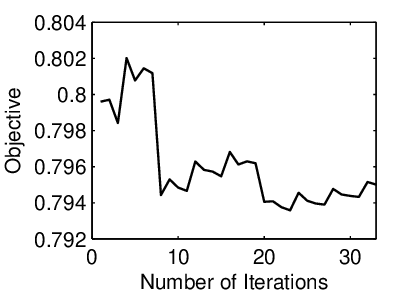} & \includegraphics[width=4.5cm]{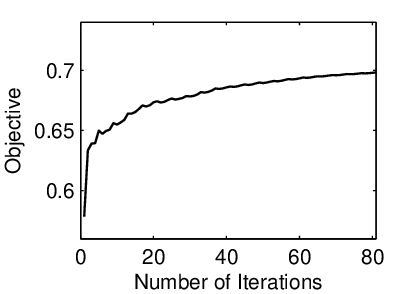} \\
\hline
  \end{tabular}
  \begin{tabular}{X{cc}X{cc}X{cc}X{cc}X{cc}}
  Initial $\alpha$ & (0;0;1;0) & (0;0;0;1) & (1;1;0;0) & (1;0;1;0) \\
  \hline
gEnM.ON &\includegraphics[width=4.5cm]{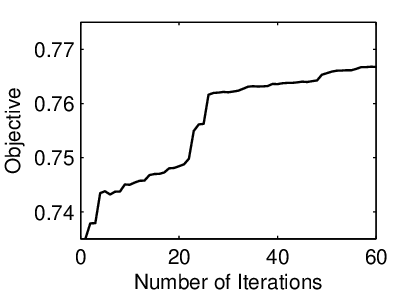}& \includegraphics[width=4.5cm]{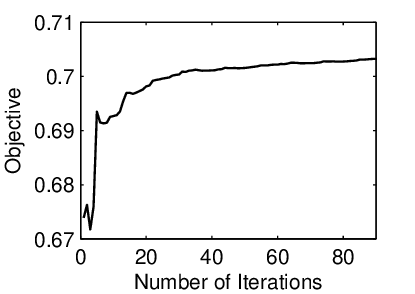} & \includegraphics[width=4.5cm]{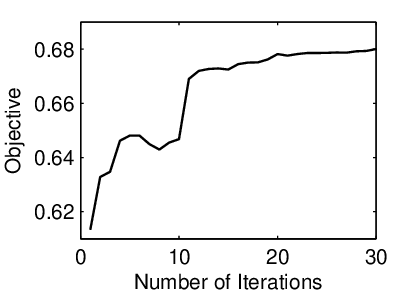} & \includegraphics[width=4.5cm]{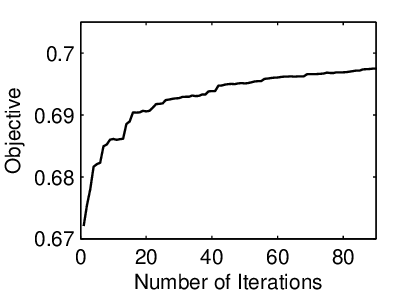}\\
UnsEnM &\includegraphics[width=4.5cm]{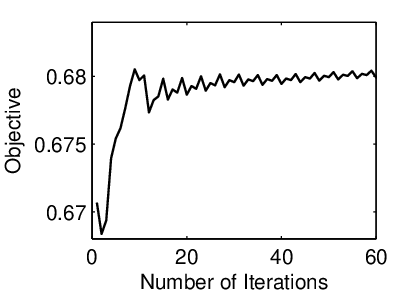}& \includegraphics[width=4.5cm]{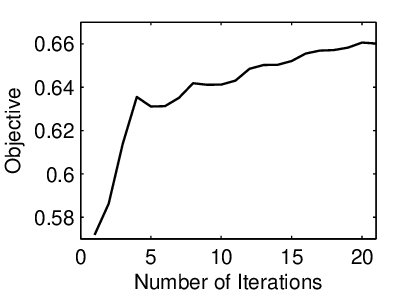} & \includegraphics[width=4.5cm]{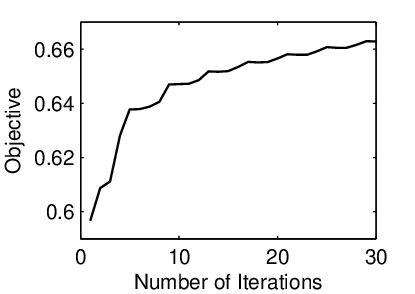} & \includegraphics[width=4.5cm]{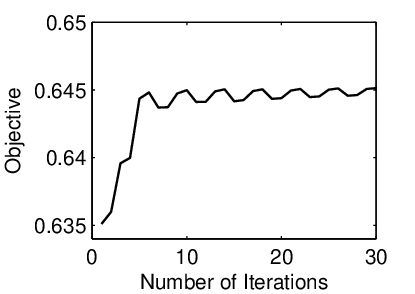} \\
\hline
  \end{tabular}
  \caption{Learning curves of EnM.ON and UnSEnM with different initial points on MED.}
  \label{fig.l.curv1}
\end{sidewaysfigure*}

\begin{sidewaysfigure*}\tiny
\ContinuedFloat
\centering
  \begin{tabular}{X{cc}X{cc}X{cc}X{cc}X{cc}}
  \hline
  Initial $\alpha$ & (1;0;0;1) & (0;1;1;0) & (0;1;0;1) & (0;0;1;1) \\
  \hline
gEnM.ON &\includegraphics[width=4.5cm]{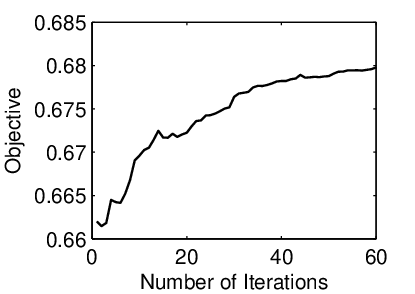}& \includegraphics[width=4.5cm]{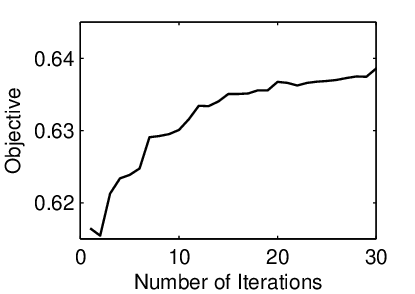} & \includegraphics[width=4.5cm]{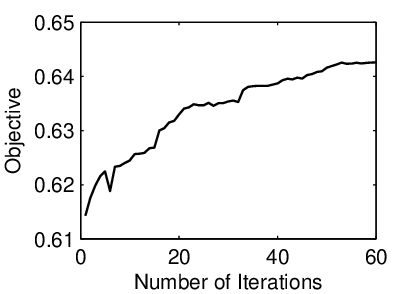} & \includegraphics[width=4.5cm]{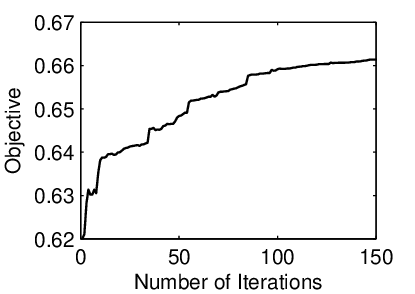}\\
UnsEnM &\includegraphics[width=4.5cm]{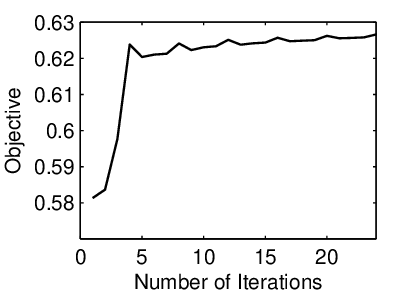}& \includegraphics[width=4.5cm]{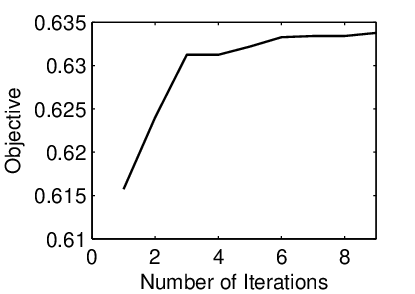} & \includegraphics[width=4.5cm]{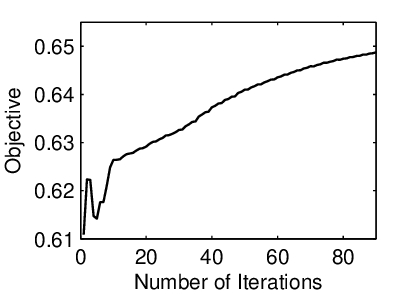} & \includegraphics[width=4.5cm]{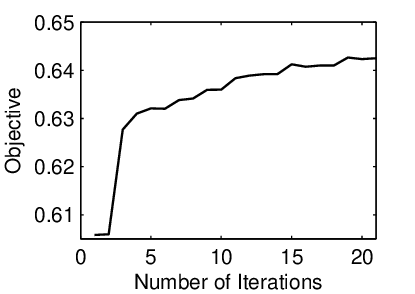} \\
\hline
  \end{tabular}
  \begin{tabular}{X{cc}X{cc}X{cc}X{cc}X{cc}}
  Initial $\alpha$ & (1;1;1;0) & (0;1;1;1) & (1;0;1;1) & (1;1;0;1) \\
  \hline
gEnM.ON &\includegraphics[width=4.5cm]{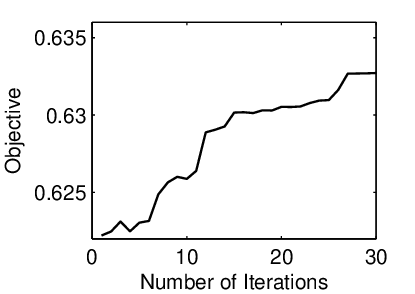}& \includegraphics[width=4.5cm]{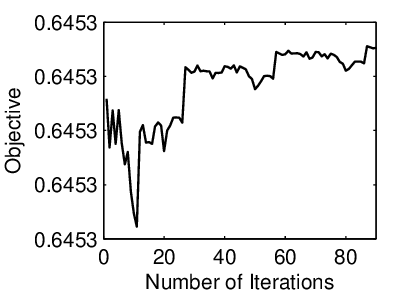} & \includegraphics[width=4.5cm]{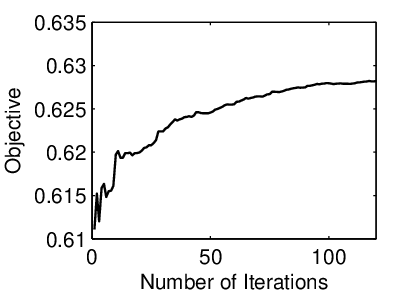} & \includegraphics[width=4.5cm]{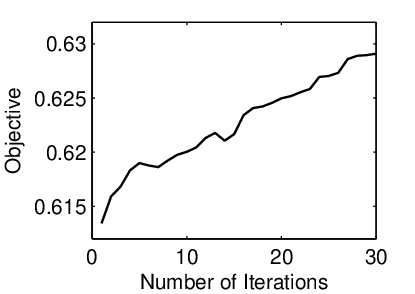}\\
UnsEnM &\includegraphics[width=4.5cm]{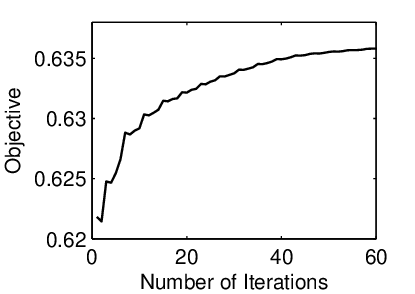}& \includegraphics[width=4.5cm]{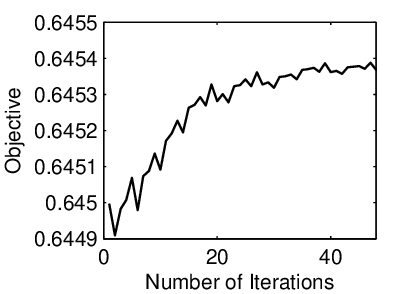} & \includegraphics[width=4.5cm]{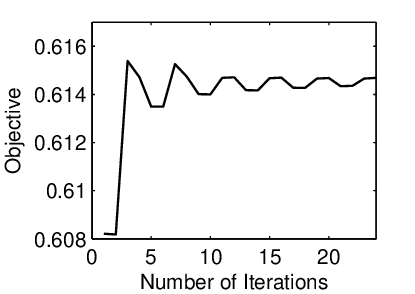} & \includegraphics[width=4.5cm]{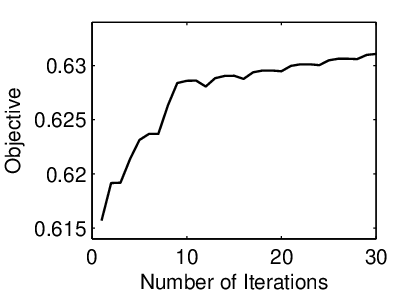} \\
\hline
  \end{tabular}
  \caption[]{Learning curves of EnM.ON and UnSEnM with different initial points on MED. (continued)}
  \label{fig.l.curv2}
\end{sidewaysfigure*}

\section{Conclusions and Discussions}
In this paper, we propose a generalized ensemble model, gEnM, which tries to find the optimal linear combination of multiple constituent rankers by directly optimizing the problem defined based on the mean average precision. In order to solve this optimization problem, the algorithms are devised in two aspects, i.e., supervised and unsupervised. In addition, two settings for the data are considered in the supervised learning, namely batch and online setting. Table \ref{tab.sum} summarises the algorithms with potential applications in practice. In brief, the gEnM.BAT can be used in those IR systems that have the knowledge of labeled data, such as academic search engines; the gEnM.ON is appropriate for real-time systems where the data is given in sequence, such as movie recommendation systems; and the UnsEnM is proposed for those systems without the knowledge of labeled data, such as search engines.

\begin{table*}
\centering
\caption{Summary of the algorithms: gEnM.BAT, gEnM.ON and UnsEnM.}
\label{tab.sum}
\begin{tabular}{lccl}
  \hline
  Algorithm & Category & Setting & Application \\
  \hline
  gEnM.BAT & supervised & batch & academic search, etc. \\
  gEnM.ON  & supervised  & online & movie recommendation, etc. \\
  UnsEnM  &  unsupervised & batch & search engine, etc. \\
  \hline
\end{tabular}
\end{table*}

An experimental study was conducted based on the public data sets. The encouraging results verify the effectiveness of the proposed algorithms for both homogeneous and heterogeneous data. The gEnM performance is always better than the EnM, except for the case of UnsEnM on CACM. Briefly, the difference between gEnM.BAT and EnM is statistically significant in most cases; the gEnM.ON performs the best among the proposed algorithms for the MED, CRAN and CACM; and the unsupervised UnsEnM is more applicable for heterogeneous data than the supervised algorithms.

While we have shown the effectiveness of the proposed algorithms, we have not yet analyzed the computational complexity of the algorithms. Though we simplified the computation of the derivative and Hessian matrix, we were unable to reduced the complexity of the batch algorithm based on Newton's method. A possible future direction is to exploit cheaper and faster algorithms for the batch setting. Another interesting research topic is the selection of initial weights, which is actually an open research issue in nonlinear programming.

Apart from the potential improvements with regard to algorithms, the selection of constituent rankers is an extremely important issue. This problem may be resolved if we can identify which ranker is redundant for the ensemble. In this paper, we use human heuristics for choosing the four rankers. However, a concrete framework to effectively evaluate the contribution of each ranker is no doubt a subject worthy of further study.



\bibliographystyle{IEEEtran}
\bibliography{IEEEabrv,enm}

\appendices
\section{Derivation of the derivative of $\Lambda'$}\label{appendix.a}
(1) Derivation of the first derivative

According to the calculus chain rule, the derivative of objective in Problem P4 with respect to $\alpha_k, k=1,2,..,K_{\phi}$ is
\begin{equation}\label{equ.g.fird}
  \frac{\partial\Lambda'}{\partial\alpha'_k}=\frac{1}{L}\sum^L_{i=1}\frac{1}
  {|D_i|}\sum^{|D_i|}_{j=1}\frac{-j\sum_{d\neq d_j}\frac{\partial g'_{ij}}{\partial\alpha'_k}}{(1+\sum_{d\neq d_j}g'_{ij})^2},
\end{equation}
where
\begin{equation}\label{equ.g}
  \frac{\partial g'_{ij}}{\partial\alpha'_k}=-\beta s_{d_j,d}(\phi_k(q_i))g'_{ij}(1-g'_{ij}).
\end{equation}
(2) Derivation of the second derivative

Also by the chain rule, the second derivative with respect to $\alpha'_l, l=1,2,..,K_{\phi}$ is
\begin{equation}\label{equ.g.secd}
 \begin{aligned}
  & \frac{\partial^2\Lambda'}{\partial\alpha'_k\partial\alpha'_l}=
  \frac{1}{L}\sum^L_{i=1}\frac{1}{|D_i|} \\
  & \sum^{|D_i|}_{j=1}\frac{-j\sum\frac{\partial^2g'_{ij}}{\partial\alpha'_k\partial\alpha'_l}
  (1+\sum g'_{ij})^2+2j\sum\frac{\partial g'_{ij}}{\partial \alpha'_k}\sum\frac{\partial g'_{ij}}{\partial \alpha'_l}
  (1+\sum g'_{ij})}
  {(1+\sum_{d\neq d_j}g'_{ij})^4},
 \end{aligned}
\end{equation}
where
\begin{equation}
  \frac{\partial^2 g'_{ij}}{\partial\alpha'_k\partial\alpha'_l}=
  -\beta s_{d_j,d}(\phi_k(q_i))(1-2g'_{ij})\frac{\partial g'_{ij}}{\partial \alpha_l},
\end{equation}
and $\frac{\partial g'_{ij}}{\partial \alpha_l}$ can be calculated by Equation \ref{equ.g}.

\section{Approximation of the derivative of sigmoid function} \label{appendix.b}

For notational simplicity, we begin by considering the following sigmoid function:
\begin{equation}\label{equ.apd.sig}
f(x)=\frac{1}{1+\exp(\beta x)}.
\end{equation}

\begin{theorem}\label{thm.sig.app}
The derivative of function (\ref{equ.apd.sig}) can be approximated as follows:
\begin{equation}\label{equ.der.app}
\frac{\partial f(x)}{\partial x}\simeq
\begin{cases}
\begin{aligned}
  &-\beta(f(x)-f^2(x)), &&\text{if } -\frac{2}{\beta}<x<\frac{2}{\beta};\\
  &0, && \text{if } x<-\frac{2}{\beta} \text{ or } x>\frac{2}{\beta}. \\
\end{aligned}
\end{cases}
\end{equation}
if the scaling constant $\beta$ is large.
\end{theorem}
\begin{proof}
We apply the centered linear approximation method to the approximation of the sigmoid function as shown in Figure \ref{fig.sig.app}, which is described below:
\begin{equation}
f(x)\simeq
\begin{cases}
\begin{aligned}
& f(x), && \text{if } -\frac{2}{\beta}<x<\frac{2}{\beta};\\
& 0, && \text{if } x<-\frac{2}{\beta}; \\
& 1, && \text{if } x>\frac{2}{\beta}.
\end{aligned}
\end{cases}
\end{equation}
Hence $f(x)(1-f(x))=0$ if $x<-\frac{2}{\beta}$ or $x>\frac{2}{\beta}$. This completes the proof.
\end{proof}

We note that this approximation is more precise with a larger $\beta$.

\begin{figure}
  \centering
  \includegraphics[width=0.5\textwidth]{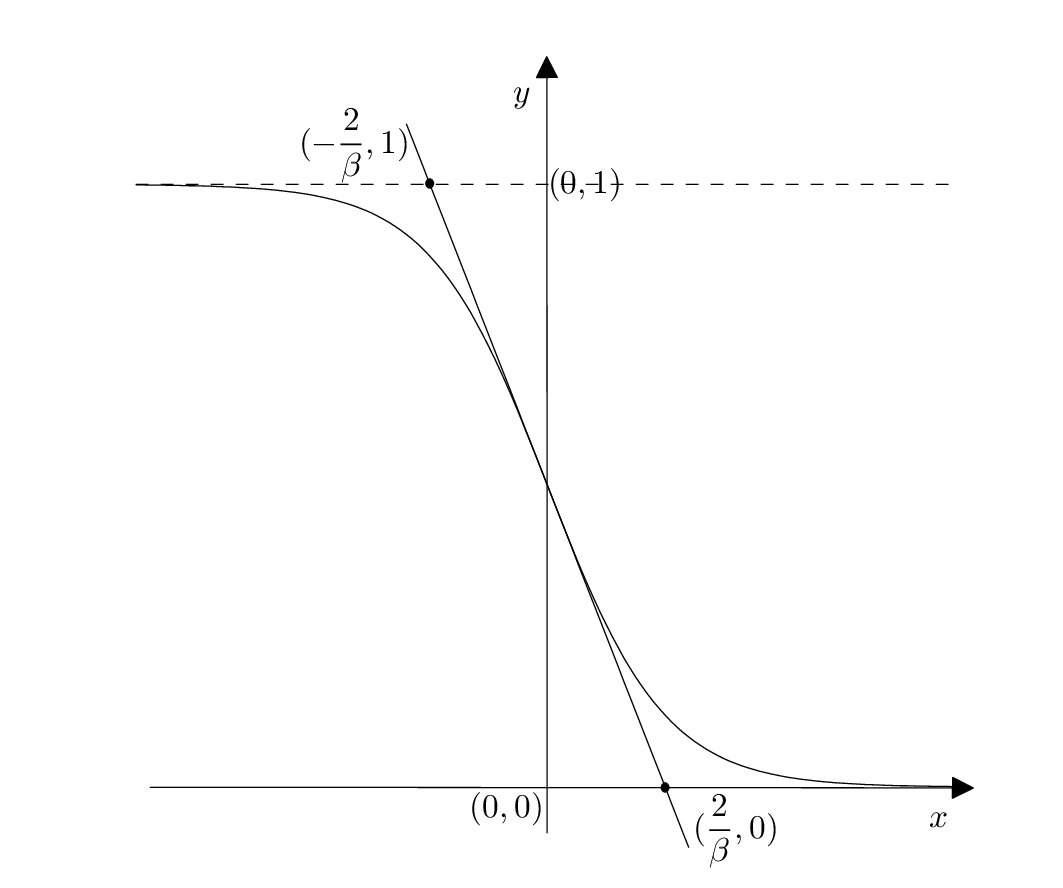}
  \caption{The approximation of sigmoid function through the centered linear approximation method. ($\beta=300$)}\label{fig.sig.app}
\end{figure}

\begin{remark}
The derivative function (\ref{equ.g}) can be approximated by:
\begin{equation}
\frac{\partial g'_{ij}}{\partial\alpha'_k}\simeq
\begin{cases}
\begin{aligned}
& -\beta s_{d_j,d}(\phi_k(q_i))g'_{ij}(1-g'_{ij}), \\
& \text{if } -\frac{2}{\beta}<\sum_k \alpha'_k s_{d_j,d}(\phi_k(q_i))<\frac{2}{\beta};\\
& 0, && \text{otherwise}.
\end{aligned}
\end{cases}
\end{equation}
if the scaling constant $\beta$ is large.
\end{remark}

\section{Proof of Lemma \ref{lemma.cond2}}

In this section, we only sketch the proof of Lemma \ref{lemma.cond2}.
\begin{proof}[Sketch of Proof]
In this proof, we use simple symbols for clarity. For example, $g(\alpha_t)$ denotes $g'_{ij}(\alpha'_t)$.
\begin{equation*}
\begin{aligned}
&\nabla f(\mathbf{x},\alpha_{t+1})^2-\nabla f(\mathbf{x},\alpha_t)^2 \\
&=\left(\frac{1}{D}\sum_{i=1}^{D} \frac{j\beta\sum s g(\alpha_{t+1})(1-g(\alpha_{t+1}))}{(1+\sum_{d\neq d_j}g(\alpha_{t+1}))^2}\right)^2 - \\
& \left(\frac{1}{D}\sum_{i=1}^{D} \frac{j\beta\sum s g(\alpha_t)(1-g(\alpha_t))}{(1+\sum_{d\neq d_j}g(\alpha_t))^2}\right)^2 \\
&< \left(\frac{1}{D}\sum_{i=1}^{D} j\beta\sum s g(\alpha_{t+1})(1-g(\alpha_{t+1}))\right)^2
\end{aligned}
\end{equation*}
For $g(\alpha_{t+1})-g(\alpha_{t+1})^2$, we have
\begin{equation*}
\begin{aligned}
g(\alpha_{t+1})-g(\alpha_{t+1})^2
&< \frac{1}{2+\exp(\beta \sum (\alpha_t+\eta\nabla f) s)} \\
&< \frac{1}{2+\exp(\beta \sum \eta\nabla f s)}.
\end{aligned}
\end{equation*}
Thus, we have
\begin{equation*}
\nabla f(\mathbf{x},\alpha_{t+1})^2-\nabla f(\mathbf{x},\alpha_t)^2
< \left(\frac{1}{D}\sum_{i=1}^{D} j\beta\sum s \frac{1}{2+\exp(\beta \sum \eta\nabla f s)}\right)^2
\end{equation*}
It is easy to show that the $\frac{1}{1+\exp(\eta)}$ is the summand of a convergent infinite sum. This result implies that $\nabla f(\mathbf{x},\alpha_t)^2$ converges because it is bounded and its oscillations are damped.
\end{proof}

\end{document}